\documentclass[11pt]{article}

\usepackage[utf8]{inputenc}
\usepackage[normalem]{ulem}

\date{}
\usepackage[left= 3 cm,right=3 cm,top=3 cm,bottom=3 cm]{geometry}
\usepackage{color}
\usepackage{graphicx}
\usepackage{xspace}
\usepackage{url}
\usepackage{amsmath,amssymb,amsthm}
\usepackage{hyperref}
\hypersetup{pdfstartview=XYZ}
\usepackage{capt-of}
\usepackage{enumitem}
\usepackage{array,multirow}

\usepackage{authblk}

\usepackage{tikz}
\usetikzlibrary{positioning}
\usetikzlibrary{patterns}
\usetikzlibrary{decorations}
\tikzset{
bicolor/.style 2 args={
  dashed,dash pattern=on 10pt off 10pt,#1,
  postaction={draw,dashed,dash pattern=on 10pt off 10pt,#2,dash phase=10pt}
  },
}

\newcommand{\commentout}[1]{}

\usepackage[vlined,english,boxed]{algorithm2e} 
      \SetAlgoLined
      \SetKwInput{Input}{Input}
      \SetKwInput{Output}{Output}
      \SetKw{Return}{Return}
      \SetKwIF{If}{ElseIf}{Else}{If}{then}{else if}{else}{endif}
      \SetKwFor{For}{For}{do}{done}
      \SetKwFor{for}{For}{}{}
      \SetKwFor{While}{While}{do}{done}
      \SetKwFor{Repeat}{Repeat}{}{}
      \SetKwRepeat{Do}{Do}{while}
      \SetKwFor{Procedure}{Procedure}{}{}
      \SetKwFor{Function}{Function}{}{}
      \SetKw{Return}{Return}%
      \SetKw{return}{return}%
      \SetKwComment{Comment}{/* }{ */}
      \SetSideCommentRight
      \DontPrintSemicolon
\usepackage{todonotes}

\long\def\jump#1\finjump{}
\long\def\beglongversion#1\endlongversion{#1}

\newtheorem{theorem}{Theorem}
\newtheorem{lemma}{Lemma}
\newtheorem{corollary}{Corollary}
\newtheorem{proposition}{Proposition}
\newtheorem{definition}{Definition}
\newtheorem{property}{Property}

\newcommand{\maxl}{\max_{\vert \mbox{\scriptsize left}}}
\newcommand{\maxr}{\max_{\vert \mbox{\scriptsize right}}}
\newcommand{\maxs}{\max_{\vert \mbox{\scriptsize sofar}}}
\newcommand{\tot}{\mbox{total}}
\newcommand{\tote}{\emph{total}}
\newcommand{\setl}{S_{\mbox{\scriptsize left}}}
\newcommand{\setr}{S_{\mbox{\scriptsize right}}}
\newcommand{\sets}{S_{\mbox{\scriptsize sofar}}}

\newcommand{\conv}{\mbox{conv}}

\DeclareMathOperator*{\argmax}{arg\,max}

\usepackage[framemethod=TikZ]{mdframed}
\mdfdefinestyle{MyFrame}{%
    linecolor=orange,
    outerlinewidth=1pt,
    roundcorner=20pt,
    innertopmargin=\baselineskip,
    innerbottommargin=\baselineskip,
    innerrightmargin=20pt,
    innerleftmargin=20pt,
    backgroundcolor=yellow!10!white
}
\usepackage{amsfonts}

\author[1]{Fariza Aklouche}
\author[2]{Pierre Bergé}
\author[3]{Michel Habib}
  
 \affil[1]{University Mouloud Mammeri of Tizi-Ouzou, Tizi-Ouzou, Algeria}
 \affil[2]{Université Grenoble Alpes, CNRS, Grenoble INP, LIG, 38000 Grenoble, France}
 \affil[3]{IRIF CNRS \& Université Paris Cité, Paris, France}
\title{On the maximum  weight convex problem for  some geometric graph-convexities\footnote{An extended abstract~\cite{cosi25} of this version was presented in conference COSI, Bejaia, Algeria, 2025}}

\begin{document}

\maketitle

\begin{abstract} 
For a given geometric graph-convexity on a graph $G$ equipped with a  weight function on the vertices with value in $\mathbb{Z}$, the \textsc{max weight convex set} problem consists in determining the convex set $S$ with maximum weight (sum of the weight of the vertices in $S$). Although the problem is NP-complete in  general, it remains polynomial for particular cases.

After a survey of known results, our main contribution uses a generalisation of the maximum subsequence problem to laminar trees. Then we derive a linear algorithm for proper interval graphs and a quadratic one for interval graphs. Both improve the state of the art.

\end{abstract}
\section{Introduction to convexity spaces and geometric convexities}

\textbf{Definitions}. A convexity space is an ordered pair $(V, \cal C)$, where $V$ is a non-empty set and $\cal C$ is a collection of $V$-subsets, to be regarded as convex sets, such that:

\begin{enumerate}[label=(\roman*)]
    \item $\emptyset$, $V \in \cal C$.
    \item Arbitrary intersections of convex sets are convex.
    \item Every nested union of convex sets is convex.
 \end{enumerate}

Given a convexity space $(V, \mathcal{C})$, the convex hull of $ X \subseteq V$ denoted by $\conv(X)$
is the smallest convex set containing $ X$, being the intersection of all the convex sets that contain $X$. Families satisfying (iii) are known as inductive systems.
A vertex $x$ of a convex set $S$ is an \emph{extreme point} of $S$ if $S\setminus \{ x \}$ is also a convex set.

Introduced first by Edelman and Jamison in 1985~\cite{edelman1985theory}, convex geometries appeared in the literature under different names and aspects, the most famous one being \emph{antimatroids}.
The convexity $\mathcal{C}$ is a \emph{convex geometry} if it satisfies the \emph{Minkowsky-Krein-Milman} property: every convex set is the convex hull of its extreme vertices.
Equivalently: a convex geometry is a convexity space satisfying the \emph{Anti-Exchange Property}: for any convex set $S$ and two distinct points $x , y \notin S$, $x \in \conv(S \cup \{y\})$ implies
$y \notin  \conv(S \cup \{x\})$. Given a convex geometry $(V, \mathcal{C})$, then the set system
$(V, \mathcal{F})$ where $\mathcal{F} = \{X: X = V\backslash Y, Y \in \mathcal{C})$ is an \emph{antimatroid}. 
We call the sets $X \in \mathcal{F}$ \emph{feasible sets}. In the remainder, we say that a convexity space is \emph{geometric} if it is a convex geometry. 

Following Duchet \cite{Duchet88},
a \emph{graph-convexity space} is an ordered pair $(G, \cal C)$, formed by a connected graph $G=(V,E)$  and a convexity $\cal C$ on $V$ such that $(V, \cal C)$ is a convexity space satisfying
the additional axiom:
\begin{enumerate}[label=(\roman*)]
\setcounter{enumi}{3}
\item Every member of $\cal C$ induces a connected subgraph of $G$.
\end{enumerate}

Antimatroids capture various elimination orderings on graph classes, which are the basis of numerous efficient algorithms. 
There is a close relationship between convexity spaces and the notion of \emph{abstract betweenness} defined by Menger~\cite{menger1928untersuchungen}. 
Betweenness is a ternary relation, that relates the ``placement" of a point $z$ \emph{between} two other points $a$ and $b$. 
We say $z$ belongs to the interval $I[a,b]$ when $z$ is \emph{between} $a$ and $b$. The notion of interval varies with every betweenness. 
As noticed by Chv\'{a}tal~\cite{Chvatal08}, several abstract betweennesses can be defined on graphs. 
Given a betweenness relation, one can deduce a graph-convexity space as follows:
$C \subseteq V(G)$ is convex if $\forall a, b \in C$ and $\forall z \in I[a, b]$, $z \in C$.

As a first example, the set of vertices $A$ is a \emph{geodesic} convex set if $A$ contains all vertices which belong to shortest paths between vertices in $A$.
Geodesic convexity is geometric for a graph $G$ if and only if (iff) $G$  is a Ptolemaic graph, {\em i.e.} chordal and distance hereditary~\cite{Howorka81}.

A well studied antimatroid is the one that rises from \emph{chordal graphs}, {\em i.e.} the graphs where the largest induced cycle is a triangle. 
Chordal graphs are characterized by a vertex ordering known as a \emph{perfect elimination ordering} (PEO). 
An ordering $\sigma = v_1, v_2, \ldots, v_n$ is a PEO for $G$ if for all $i\in[n], v_i$ is \emph{simplicial} in $G[v_1, \ldots, v_{i-1}]$, meaning the neighbourhood of $v_i$ to its left in $\sigma$ induces a clique. Thus $v_n$ is a simplicial vertex in $G$. 
The set system $(V, \mathcal{F})$ whose ground set are the vertices of a chordal graph, and its feasible sets are the suffixes of PEOs form an antimatroid~\cite{korte2012greedoids}. 
Given such an antimatroid, the corresponding convex geometry is the tuple $(V, \mathcal{N})$ where $S \subseteq V$ is convex if all chordless paths between the vertices in $S$ are also in $S$. 
That is, for all $a,c \in S$, if $b$ lies in a chordless $(a,c)$-path in $G$, then $b \in S$. 
This is known as the \emph{monophonic convexity} or simply m-convexity. 
Monophonic convexity yields a convex geometry on a graph $G$ iff $G$ is weak polarizable (every subgraph is chordal or it has a proper homogeneous set)~\cite{DraganNB99}. Naturally this class of graphs contains all chordal graphs.

Another well studied antimatroid is the double shelling antimatroid of posets. 
Let $P(V, \prec)$ be a poset. 
The double shelling antimatroid on $P$ is the set system whose feasible sets are unions of ideals and filters of $P$. 
The corresponding convex geometry is a set system $(V,\mathcal{N})$, where a set $S \in \mathcal{N}$ is convex if for all $a,c \in V$, every $b$ that satisfies $a \prec b \prec c$ or $c \prec b \prec a$ is also in $S$. 
For a more comprehensive survey on this topic, we refer the reader to the monograph \emph{Greedoids} by Korte, Lov{\'a}sz, and Schrader~\cite{korte2012greedoids}, and to~\cite{FarberJ87, Duchet88} for more on graph convexities.

As considered by Chvatal  \cite{Chvatal08}, many other discrete convexities in graphs can be defined via intervals. As an example, let us mention the \emph{$m_3$ convexity} in which
$I[x,y]=\{z ~|~ z \in p$ an induced path of length $\geq 3$ from $x$ to $y\}$.
If $G$ is HDD-free (House, Domino, Diamond), then the $m_3$-convexity is a convex geometry \cite{FarberJ87}.

Furthermore, the \emph{2-paths-convexity}, used for the study of cocomparability graphs ({\em i.e.} complement of comparability graphs) and asteroidal triple free graphs (AT-free graphs for short) is defined for the following notion of interval:
$I[x,y]=\{z |\; \exists p, q$ two chordless paths, such that $p$ from $z$ to $x$ avoids $N[y]$ and $q$ from $z$ to $y$  avoids $N[x]\}$.
If $G$ is a cocomparability graph, the 2-paths convexity is geometric.

Following \cite{KW71}, a convexity space is said to have \emph{Carath\'eodory number} $d$
iff $d$ is the smallest positive integer with the following property:
if a point lies in the convex hull of a set $X$,
then it lies in the convex hull of a subset $X_0 \subseteq X$ such that $\vert X_0\vert \leq d$.

\medskip

\textbf{Problem studied}. In this paper, we mainly deal with discrete convex geometry on which we consider the simple optimization problem proposed below. It has applications to polytopes \cite{abam2022maximum,Doignon16}
and  has already gave rise to very nice combinatorial algorithms. We will consider  some of them here.

\begin{mdframed}[style=MyFrame]
\textbf{Name :} \textsc{max weight convex set}  

\textbf{Data:} ground set $V$, equipped with a weight function
 $\omega $:   $V \rightarrow \mathbb{Z}$  and a discrete convex geometry $\cal C$ on $V$.

\textbf{Result:}  $S$ a ${\cal C}$-convex 
set in $V$ with maximum weight\footnote{
It should be noted that the weight of the empty set is $0$, which could be the optimum if for every vertex $x \in V$, $\omega(x) <0$.}: $\omega(S)=\Sigma_{x \in S} \omega(x)$.
\label{pb:max_convex}
\end{mdframed}

Unfortunately, this problem is NP-hard in general using an easy 
reduction from \textsc{independent set}~\cite{eppstein1992finding}. Therefore it is very interesting to consider graph classes on which this problem can be solved efficiently.
If a graph $G$  admits several connected components
$G_1, \dots G_k$  for every graph convexity $\cal C$, one can consider the problem on each connected component and the value for $G$ is simply the maximum of the computations.
Therefore in what follows all our graphs are supposed to be connected.

The difficulty of the problem comes from the closure operation associated with $\cal C$. Suppose we have a subset $A \subseteq V(G)$  with a very positive weight, to go from $A$ to $\conv(A)$ we could be forced to add vertices with negative weights. Another difficulty comes from the number of convex sets to be considered. In some graph classes, a given graph $G$ may admit  an exponential number of convex sets.

\section{Maximum-weighted convex sets on graphs}

\textbf{Graph notation}. From now on, the graphs considered are finite, simple (no multiple edge), loopless, undirected, and connected. 
For such a graph $G=(V, E)$, we will denote by $n=|V|$ the number of vertices and $m=|E|$ the number of edges.
For a vertex $x \in V$, $N(x)$ denotes its neighbourhood.
The distance between two vertices $u$ and $v$ of $G$ is the length of the shortest $(u,v)$-path and denoted by $d(u,v)$.
A \emph{module} $M \subseteq V$ is a collection of vertices that have the same neighbourhood outside of $M$: for all $x,y \in M: N(x) \setminus M = N(y) \setminus M$. In particular,
$\emptyset$, every singleton $\{x\}$ and the whole vertex set $V$ are called \emph{trivial} modules.
A \emph{prime} graph is a graph that admits only trivial modules.
As a particular case of module we have \emph{twins}. Two vertices $x,y$ are \emph{twins} if $N(x) \setminus \{y\} = N(y) \setminus \{x\}$. 
If, in addition, $xy \in E$, we say $x,y$ are \emph{true twins}, otherwise they are \emph{false twins}. 
We say two vertex sets $X$ and $Y$ \emph{overlap} if and only if (iff) $X\cap Y \neq \emptyset$ , $X\setminus Y \neq \emptyset$ and $Y \setminus X \neq \emptyset$.  

\textbf{Known algorithms}. We focus now on efficient algorithms for some classes of graphs such as split graphs, threshold graphs,  proper interval graphs and interval graphs and improve the known algorithms. The \textsc{max weight convex set} problem has been considered for arrays under the name \textsc{maximum sum subsequence} also called \textsc{maximum subvector problem}. A nice linear time algorithm is proposed in \cite{Bentley84}. In this case the convex sets are just the contiguous subsequences. More algorithms have been proposed for other cases, as listed below.

\begin{itemize}
\item For trees with geodesic convexity, the problem can be solved in $O(n^2)$ using dynamic programming  \cite{eppstein1992finding}, improved to $O(n\log n)$ in \cite{CarlsonE06}.
In this case the geodesic convex sets are just the subtrees. We wonder if a linear time algorithm could be obtained.

\item For partial orders when convex sets are just intervals of the partial order, the problem can easily be solved in quadratic time. But when convex sets are the ideals of the partial orders then one can use a reduction to maximum flow problems \cite{Picard75}. This gives formally:
for a poset $P=(X, \leq_P)$, and a weight function $w : X \rightarrow \mathbb{Z}$
finding an ideal of $P$ that minimizes $w$ can be done in $O(|X|m\log \frac{|X|^2}{m})$ time, where
$m$ is the number of cover relations in $P$.
But several linear time algorithms already exist for computing a maximum flow in graphs \cite{chen2022maximumflowminimumcostflow} and therefore a linear running time can be proposed.

\item Polynomial-time can be achieved for some well-structured graph classes such as split and chordal graphs \cite{MerckxCD16,CardinalDM17,CardinalDM19, Merckx19}.
For chordal graphs they use a data structure, namely the  "Clique-Separator Graph" introduced in \cite{IBARRA2009} and then apply the partial order algorithm. The bottleneck of this approach is given by construction of the data structure in $O(n^3)$. 
With the same remark as above for maximum flow the complexities of their algorithms can be improved.

\item 
On  $\mathbb{R}^d$ equipped with its usual convex geometry, the problem is polynomial for $d\leq 2$  and NP-hard for $d\geq 3$~\cite{abam2022maximum}.
\end{itemize}

Observe that a closely related problem has been intensively studied, namely the \textsc{maximum-weight connected subgraph}, see for example \cite{El-KebirK14}. This generalization is NP-hard and furthermore not related to convexity since the intersection of two connected subgraphs of a given graph is not necessarily connected.

\textbf{Summary}. The following table summarizes both the literature and our results about the  complexity of computing the maximum-weight convex set when $G=(V,E)$ a graph with $|V(G)|=n$ and $|E(G)|= m$. For each class of graphs tackled, we focus on the natural convex geometry associated with its structure (details are given in each section). We also mention the case of trees which stand in our opinion as an important challenge for future research.

\begin{table}[h]
\centering
\begin{tabular}{|c|c|c|}
  \hline
  Graph classes / Discrete structures & Known results & \textbf{Our results} \\
  \hline
   \hline
  Trees   & $O(n\log n)$ \cite{CarlsonE06} & \\
  \hline
 PQ-Trees &  & $ O(n^2)$ Th.~\ref{PQ-tree} \\
 \hline
 \hline
 Chordal graphs & \multirow{4}{*}{$O(n^2m^2\log \frac{n^2}{m})$ \cite{CardinalDM19}} & \\\cline{1-1}\cline{3-3}
  Proper interval graphs & & $ O(n)$ Th.~\ref{properinterval} \\
  \cline{1-1}\cline{3-3}
  Interval graphs &  & $O(n^{2})$ Th.~\ref{intervalgr} \\
  \cline{1-1}\cline{3-3}
  Threshold graphs & & $O(n+m) $ Th.~\ref{Threshold}\\
  
  \hline

\end{tabular}
\caption{\textsc{max weight convex set} on different structures.}
\end{table}


\section{Maximum sequences of leaves in a laminar tree}

Let us now introduce a tool, namely \emph{laminar trees}  which will be very useful for our algorithms on \textsc{max weight convex set} for well-structured classes of graphs.
A natural generalization of the \textsc{maximum sum subsequence} problem in an array is to consider a similar problem on laminar trees. 

\subsection{Laminar families and trees}

As defined in \cite{Schrijver}, a \emph{laminar family} on a ground set $V$ is a subset ${\cal F} \subseteq 2^V$ such that
for all $A, B \in \cal F$, either $A \subseteq B$ or  $B \subseteq A$ or $A\cap B=\emptyset$. Observe a laminar family $\cal F$ on $V$ is naturally represented by a rooted forest, denoted $T_{\cal F}$ and called \emph{${\cal F}$-laminar forest}, such that every set $A\in\cal F$ admits a representative node $u_A$ in $T_{\cal F}$. The leaves of the subtree rooted at $u_A$ provides us with the content of set $A \in \cal F$. Note that if $X\in\cal F$, then $T_{\cal F}$ is a rooted tree. More generally, when $T_{\cal F}$ is connected, we say that it is a \emph{laminar tree}.

Given a planar embedding of a laminar tree $T_{\cal F}$ with leaves at the bottom, {\em i.e.} a \emph{representation}, the left-right reading of the leaves of $T_{\cal F}$ gives a total ordering of $X$.
In order to manipulate sets of representations of a laminar tree in a compact way, let us first introduce these simple types on the internal nodes of \emph{typed-laminar trees}.

\begin{itemize}
    \item P-node, for parallel. Any ordering of the children of a P-node is valid.
    \item Q-node. Only one ordering of the children and the reverse one are valid.
    \item R-node for rigid (uniquely ordered). A unique ordering of the children of a R-node is valid. 
    \item C-node for circular. Given a total order $\tau$ of the children of a C-node, any cyclic ordering of $\tau$ is valid.

\end{itemize}

A typed-laminar tree provides a succint encoding of a set of planar representations of the associated tree, in fact a set of permutations of its leaves, {\em i.e.} set $X$.
A set of leaves $S$ of a typed-laminar tree $T$ is said convex if the element of $S$ are contiguous in some valid  representation of $T$. 
Now we consider a weight function on $X$ and ask the maximal convex set in $T$.

\begin{mdframed}[style=MyFrame]
\textbf{Name :} \textsc{max weight laminar convex set}

\textbf{Data:} A typed-laminar tree $T$ on ground set $X$, equipped with a weight function on its leaves $\omega : X \rightarrow \mathbb{Z}$.

\textbf{Result:}  $S$ a sequence of contiguous leaves in a valid planar representation of $T$ with maximum weight.
\end{mdframed}

Observe that if all nodes of the typed-laminar tree are typed R, then we are back to the \textsc{maximum sum subsequence} problem.
Our generalization contains the well-known PQ-trees  used for interval graphs \cite{BoothL76} and also the PC-trees  used for planarity and circular-arc graphs recognition~\cite{shih1999new}. 
Let us consider the following first case as an appetizer.

\begin{proposition}\label{one-level}
There exists a linear time algorithm solving \textsc{max weight laminar convex set} on typed-laminar trees with height 1.    
\end{proposition}
\begin{proof}
 If the type of the unique internal node is P for parallel, then the solution is just the sum of the elements with positive weight.

For R and Q-nodes, we are back to the \textsc{maximum sum subsequence} problem already mentioned. Whence we can use a nice one-sweep linear time algorithm.

For a C-node (circular) we generalize the linear  time algorithm for \textsc{maximum sum subsequence} as follows. Let us denote the elements of the circular node by $a_1,\dots    a_n$. The idea is to double these elements in an array $A^+$ of size $2n$ containing $a_1, \dots a_n, a_1, \dots a_n$. While exploring this array, we maintain two variables $\sets(A^+)$ and $\setr(A^+)$ which contain respectively the maximum-weighted contiguous subset of size $\le n$ the explored area and the maximum-weighted contiguous subset of size $\le n$ finishing at the current element. In particular, at each iteration, $\sets(A^+)$ consists in selecting the best option between $\setr(A^+)$ and the set $\sets(A^+)$ computed at the last iteration. Variables $\maxs(A^+),\maxr(A^+)$ simply store respectively the cardinality of sets $\sets(A^+),\setr(A^+)$.

\begin{algorithm}[h!]
\KwIn{an array $A^+$ with $2n$ elements}
\KwOut{Values $\maxs(A^+),\maxr(A^+)$ and corresponding convex sets $\setr(A^+),\sets(A^+)$}

\BlankLine
\nl $\maxr, \maxs \leftarrow 0$; $\setr, \sets \leftarrow \emptyset$;

\nl \For{$i=1$ to $2n$}
{
\nl \eIf{$|\setr(A^+)| =n$}{$\maxr \leftarrow \max \{\maxr +A^+[i] -A^+[i-n], 0\}$} 
{$\maxr \leftarrow \max \{\maxr +A^+[i], 0\}$}

\nl update $\setr$; \\
\nl $\maxs \leftarrow  \max \{\maxs, \maxr\}$; \\
\nl update $\sets$;
}
\caption{For one internal circular node} \label{circular}
\end{algorithm}
 
After the execution of this Algorithm \ref{circular}, the set $\sets(A^+)$ contains a subsequence of maximum weight among all subsequences of size $\leq n$.
With both extremities of $\sets$ we can adjust the initial data to find the best circular permutation.
The time complexity of this dynamic programmic procedure depends linearly on the size of array $A^+$.
\end{proof}

 \subsection{PQ-trees}
 
PQ-trees are well-known typed-laminar trees  with only two types of nodes P and Q. A PQ-tree is a data structure introduced by Booth and Lueker \cite{BoothL76}. PQ-trees can be used to represent the permutations of a set $U$ in which various subsets of $U$ occur consecutively. Efficient algorithms using PQ-trees
are given in \cite{BoothL76} for recognizing interval graphs and testing graph planarity. They are used for representing many combinatorial families as evoked in \cite{Crespelle07}.

Let $\omega : \Gamma(T)\longrightarrow \mathbb{Z}$ be a weight function associating an integer with each leaf of some PQ-tree $T$. The objective is to identify the PQ-convex set $S$ which maximizes $\displaystyle{\omega(S)=\sum_{u\in S}\omega(u)}$.

\begin{mdframed}[style=MyFrame]
\textbf{Name :} \textsc{Max weight PQ-convex set}

\textbf{Data:} PQ-tree $T$ with $n$ leaves $\Gamma(T)$, a weight function $\omega : \Gamma(T)\longrightarrow \mathbb{Z}$.

\textbf{Result:}  The PQ-convex set $S$ which maximizes $\omega(S)$.
\end{mdframed}

Let $\Gamma(T)=\{u_{1},u_{2},...,u_{n}\}$ be the set of leaves of $T$, its inner nodes are either P-nodes or Q-nodes. 
The frontier $F(T)$ from the PQ-tree is the permutation of $\Gamma(T)$ which is given by the ordering of the leaves of $T$ from left to right. A PQ-tree $T'$ is equivalent to $T$ ($T \sim T'$) if we can transform the frontier of $T'$ and obtain $F(T)$ by applying the following rules : (i) permute arbitrarily the children of a $P$-node and (ii) reverse or not the ordering of the  children of a $Q$-node. A PQ-convex set $S \subseteq \Gamma (T)$ is a set of leaves which can be defined with an equivalent PQ-tree $T'\sim T$   and two leaves  $u_{i}$, $u_{j}$ : $S$ contains $u_{i}$ and $u_{j}$ and all leaves that are between $u_{i}$ and $u_{j}$ in permutation $F(T')$. Moreover, any singleton of $\Gamma (T)$ is also considered as a PQ-convex.

We propose a natural dynamic programming scheme to handle this problem in quadratic time.

\begin{theorem}\label{PQ-tree}
The maximum-weight PQ-convex set can be computed in $O(n^2)$ on a PQ-tree with $n$ vertices.
\end{theorem}

\begin{proof}
Using PQ-trees we propose an algorithm solving the problem in $O(n^{2})$. We proceed inductively. We initialize four variables for each vertex of $T$, that we update while we go up in the tree. Concretely, each variable is a function over PQ-trees and each variable for $v$ must be understood as the value of this function applied on the subtree rooted in $v$. Given some vertex $v$, by induction hypothesis, we will assume the variables have been computed for all subtrees rooted in the children of $v$ and then show how to compute these variables for $v$. Variables for leaves can be computed trivially.\\

\textbf{Variables.} First, we define formally the four variables for some PQ-tree $T$.
\begin{itemize}
    \item $\maxl(T)$ : maximum weight of a PQ-convex $S$ which is a prefix of some permutation $F(T')$, $T'\sim T$. The associated PQ-convex set is denoted by $\setl(T)$.
    \item $\maxr(T)$ : maximum weight of a PQ-convex $S$ which is a suffix of some permutation $F(T')$, $T'\sim T$. The associated PQ-convex set is denoted by $\setr(T)$.
    \item $\maxs(T)$ : maximum weight of a PQ-convex $S$ of $T$. The associated PQ-convex set is denoted by $\sets(T)$.
    \item $\tot(T)$ : the sum of the weight of all leaves in $T$. Formally, $ \tot(T)=\sum\limits_{u\in \Gamma(T)}\omega(u)$.
\end{itemize}

In brief, set $\sets(T)$ corresponds to the result we are looking for, if $T$ is the input PQ-tree.\\

\textbf{Base case :} Suppose the tree $T$ is made up of a single node $u$. If $\omega(u) \ge 0$, we fix all variables with value $\omega(u)$ and also $\setl(T)=\setr(T)=\sets(T)=\{u\}$. 

Otherwise, we fix $\tot(T) = \omega(u)$ but all other variables equal to $0$. Furthermore, $\setl(T)=\setr(T)=\sets(T)=\emptyset$.\\

\textbf{Induction step :} We distinguish two cases: the current root $v$ is either a $P$-node or a $Q$-node.

$\bullet$ \textit{$P$-node.} We suppose that the root of $T$ is a $P$-node and the variables of the subtrees below this root have been computed. Let us denote these PQ-trees by $A_{1},...,A_{k}$. Algorithm \ref{P_interval} shows the procedure to determine
the variables of $T$ when the root is a  $P$-node. Below we give both a description and a justification of this algorithm.

The computation of total$(T)$ is trivial: $\tot(T)$=$\sum_{A_{i}} \tot(A_{i})$. The same formulation holds if the root is a $Q$-node.

For the computation of $\maxl(T)$, we distinguish two cases. First suppose that $\setl(T)$ only contains elements of sets $A_{i}$ satisfying $\tot(A_{i})>0$. In this case, $\setl(T)$, which is the maximum-weighted convex set which must appear on the left of the permutation, is necessarily made up of the concatenation of all sets $A_{j}$ with $\tot(A_{j})> 0$ with a last one $\setl(A_{i})$. Indeed, the positivity of all $\tot(A_{j})$ makes our weight increase and $\setl(A_{i})$ is a maximum-weighted convex suffix. The set  $A_{i}$ which must be selected is the maximum over $\maxl(A_{i}) - \tot(A_{i})$, otherwise it could be replaced by the maximum-weighted convex suffix of another set. In brief, we select the set $A_i$ maximizing $\maxl(A_{i}) - \tot(A_{i})$ such that $\tot(A_i) > 0$, and we have:
\begin{equation}\label{eq1}
    \textstyle \maxl(T)=\maxl(A_{i})+\sum\limits_{j\neq i , {\scriptsize \tot(A_{j})}> 0} \tot(A_{j})
\end{equation}
\indent Second, we consider the case where there can be some elements of $A_{i}$ in $\maxl(T)$ such that $\tot(A_{i})\leq 0$. There cannot be two sets $A_{i}$ like that, otherwise one of them is included entirely in $\setl(T)$ and this is a contradiction as it makes the total weight decrease. So only one of these sets exists and it will be the rightmost one among the convex set: we denote it by $A_{i}$. 
All sets $A_{j}$ with $\tot(A_{j}) > 0$ have to be added to the PQ-convex to maximize its weight. In this case, $\setl(T)$ is the union of the maximum-weighted $\setl(A_i)$, where $\tot(A_{i})\leq 0$, with all sets $A_{j}$ satisfying  $\tot(A_{j}) > 0$.\\
\begin{equation}\label{eq2}
  \textstyle \maxl(T)= \maxl(A_{i})+\sum _{{\scriptsize \tot( A_{j})}> 0}  \tot(A_{j})
\end{equation}

Eventually, it suffices to compare the values obtained in both cases (Equations~\eqref{eq1} and~\eqref{eq2}) and take the maximum. The execution time for this procedure is linear in $k$. This is described in line~\ref{line:compute_maxlr} of Algorithm \ref{P_interval}.
Note that the computation of $\maxr(T)$ is symmetric and can be deduced from the rules presented for $\maxl(T)$.

\begin{figure}[h]
\centering
\scalebox{0.9}{\begin{tikzpicture}

\coordinate (p) at (4,8) {};
\coordinate (s1) at (0,6) {};
\coordinate (s2) at (1.5,6) {};
\coordinate (s3) at (3,6) {};
\coordinate (s4) at (4.5,6) {};
\coordinate (s5) at (6,6) {};
\coordinate (s6) at (7.5,6) {};
\coordinate (s7) at (9,6) {};
\coordinate (s8) at (11,6) {};

\coordinate (s11) at (-1,3) {};
\coordinate (s12) at (1,3) {};
\coordinate (s31) at (1.5,3) {};
\coordinate (s32) at (3.3,3) {};
\coordinate (s41) at (3.5,3) {};
\coordinate (s42) at (5.1,3) {};
\coordinate (s51) at (5.3,3) {};
\coordinate (s52) at (6.9,3) {};
\coordinate (s61) at (7.1,3) {};
\coordinate (s62) at (8.9,3) {};
\coordinate (s81) at (10.2,3) {};
\coordinate (s82) at (12.2,3) {};


\draw[rounded corners, color = black, line width = 0.8pt] (s1) -- (s11) -- (s12) -- (s1);
\draw[rounded corners, color = black, line width = 0.8pt] (s3) -- (s31) -- (s32) -- (s3);
\draw[rounded corners, color = black, line width = 0.8pt] (s4) -- (s41) -- (s42) -- (s4);
\draw[rounded corners, color = black, line width = 0.8pt] (s5) -- (s51) -- (s52) -- (s5);
\draw[rounded corners, color = black, line width = 0.8pt] (s6) -- (s61) -- (s62) -- (s6);
\draw[rounded corners, color = black, line width = 0.8pt] (s8) -- (s81) -- (s82) -- (s8);


\node[draw, circle, minimum height=0.2cm, minimum width=0.2cm, fill=black] at (p) {};
\draw (p) node[above,scale=1.25,yshift = 2mm] {P-node};
\node[scale = 2] at (s2) {$\ldots$};
\node[draw, circle, minimum height=0.2cm, minimum width=0.2cm, fill=black] at (s1) {};
\node[draw, circle, minimum height=0.2cm, minimum width=0.2cm, fill=black] at (s3) {};
\node[draw, circle, minimum height=0.2cm, minimum width=0.2cm, fill=black] at (s4) {};
\node[draw, circle, minimum height=0.2cm, minimum width=0.2cm, fill=black] at (s5) {};
\node[draw, circle, minimum height=0.2cm, minimum width=0.2cm, fill=black] at (s6) {};
\node[scale = 2] at (s7) {$\ldots$};
\node[draw, circle, minimum height=0.2cm, minimum width=0.2cm, fill=black] at (s8) {};

\foreach \I in {1,...,2} {
    \coordinate (v\I) at (1.4+0.4*\I,3.2) {};
    \node[draw, circle, fill=black, scale = 0.5] at (v\I) {};
}
\foreach \I in {3,...,4} {
    \coordinate (v\I) at (1.4+0.4*\I,3.2) {};
    \node[draw, circle, fill=red, scale = 0.5] at (v\I) {};
}
\foreach \I in {1,...,3} {
    \coordinate (w\I) at (3.5+0.4*\I,3.2) {};
    \node[draw, circle, fill=green, scale = 0.5] at (w\I) {};
}
\foreach \I in {1,...,4} {
    \coordinate (x\I) at (5.2+0.35*\I,3.2) {};
    \node[draw, circle, fill=green, scale = 0.5] at (x\I) {};
}
\foreach \I in {1,...,3} {
    \coordinate (y\I) at (7+0.4*\I,3.2) {};
    \node[draw, circle, fill=blue, scale = 0.5] at (y\I) {};
}
\coordinate (y4) at (8.6,3.2) {};
\node[draw, circle, fill=black, scale = 0.5] at (y4) {};


\node[scale=1.2] at (2.8,5) {$A_{r}$};
\node[scale=1.2] at (7.7,5) {$A_{\ell}$};

\node[scale=1.2, color = red] at (2.8,2.5) {$\setr(A_{r})$};
\node[scale=1.2, color = blue] at (7.8,2.5) {$\setl(A_{\ell})$};

\draw[green, line width = 1.5pt, <->] (3.9,2.5) to node[below,text width=30mm, text centered]{all subtrees with positive $\tot$} (6.6,2.5);


 \draw[line width = 1.4pt] (p) -- (s1);
\foreach \I in {3,...,6} {
     \draw[line width = 1.4pt] (p) -- (s\I);
}
 \draw[line width = 1.4pt] (p) -- (s8);




\end{tikzpicture}}
\caption{Possible structure of a convex set $\sets(T)$ (colored leaves) below a P-node}
\label{fig:Pnode}
\end{figure}
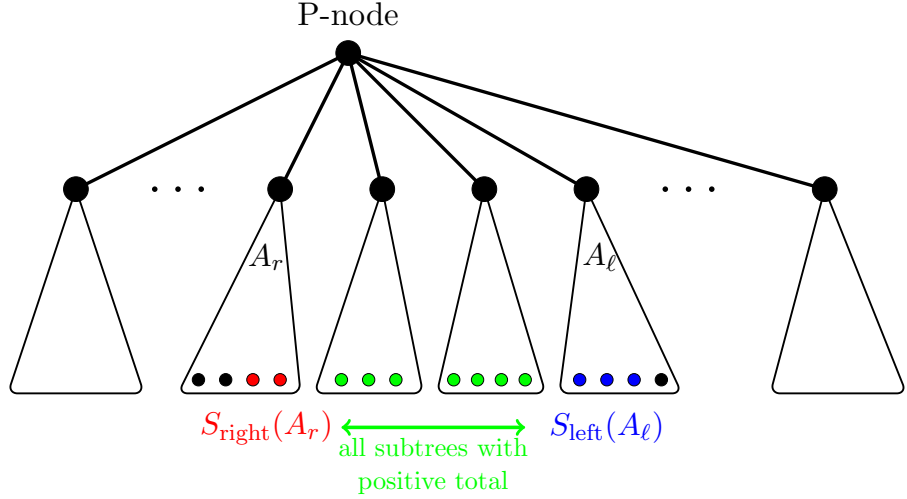

For the computation of  $\maxs(T)$, we distinguish two cases. First, suppose that the elements of $\sets(T)$ come from a single set $A_{s}$ then,
\begin{equation}\label{eq3}
   \textstyle \maxs(T)=\max\limits_{A_{s}} \maxs(A_{s}) 
\end{equation}

\begin{algorithm}[h!]
\KwIn{A PQ-tree with root $P$ and subtrees $A_{1},...,A_{k}$ }
\KwOut{Values $\maxl(T),\maxr(T),\maxs(T),\tot(T)$ and corresponding convex sets $\setl(T),\setr(T),\sets(T)$}

\BlankLine
\nl  $x_{P}^{-}\;\leftarrow \;\max\limits_{\tot( A_{i})\leq 0}\maxl(A_{i});\;A_{\ell^{-}}\leftarrow$ corresponding set;
\nl  $y_{P}^{-}\;\leftarrow \;\max\limits_{\tot( A_{i})\leq 0}\maxr(A_{i});\;A_{r^{-}}\leftarrow$ corresponding set; 
\nl  $x_{P}^{+}\;\leftarrow \;\max\limits_{\tot( A_{i})> 0}(\maxl(A_{i})-\tot ( A_{i}));\;A_{\ell^{+}}\leftarrow$ corresponding set;
\nl  $y_{P}^{+}\;\leftarrow \;\max\limits_{\tot ( A_{i})> 0}(\maxr(A_{i})-\tot ( A_{i}));\;A_{r^{+}}\leftarrow$ corresponding set;
\nl $z_{P}\;\leftarrow \;\max\limits_{A_{i}}\maxs(A_{i});\;A_{s}\leftarrow$ corresponding set;

\nl  $\tot^{+}(T)\;\leftarrow \;\sum\limits_{\tot ( A_{i})> 0}\tot ( A_{i})$; 

 \# Computation of  $\tot( T)$ \# 

\nl $\tot(T)\;\leftarrow \;\sum\limits_{A_{i}}\tot ( A_{i})$;

\# Computation of  $\maxl(T)$ and $\maxr(T)$ \#

\nl $x_{P}\;\leftarrow \;\max \{x_{P}^{-},x_{P}^{+}\};\;A_{\ell}\leftarrow $ corresponding set;

\nl $y_{P}\;\leftarrow \;\max \{y_{P}^{-},y_{P}^{+}\};\;A_{r}\leftarrow $ corresponding set;
\nl $\maxl(T)\;\leftarrow x_{P} + \tot^{+}(T);\maxr (T)\;\leftarrow y_{P} + \tot^{+}(T) $; \label{line:compute_maxlr}

\nl \eIf{$\tote(A_{\ell})\leq 0$} 
{
\nl $\setl(T) \;\leftarrow \setl
           (A_{\ell})\cup \{A_{j}\;: \tot(A_{j})> 0\}$;
  }{
\nl $\setl(T) \;\leftarrow \setl
            (A_{\ell})\cup \{A_{j}\;:\;j\neq \ell, \tot(A_{j})> 0\}$; \label{line:left_neg}
}
\nl \eIf{$\tote(A_{r})\leq 0$} 
{
\nl $\setr(T) \;\leftarrow \setr(A_{r})\cup \{A_{j}\;: \tot(A_{j})> 0\}$;
}{
\nl $\setr(T) \;\leftarrow \setr(A_{r})\cup \{A_{j}\;:\;j\neq r, \tot(A_{j})> 0\}$;
}

 \#  Computation of $\maxs(T)$ \# 

\nl $(A_{L},A_{R})\;\leftarrow \;\argmax\limits_{A_{i},A_{j},i \neq j} ((\maxl(A_{i})-\max\{\tot(A_{i}),0\}),(\maxr(A_{j})- \max \{\tot (A_{j}),0\}))$; \label{line:enum_pair}

\nl $\hat{x}_{P}\;\leftarrow \;\maxl(A_{L})$ ; $\hat{y}_{P}\;\leftarrow \;\maxr(A_{R})$ ; 

\nl  \eIf{  $z_{P}< \hat{x}_{P}+\hat{y}_{P}+\tote^{+}(T)$ }
{
\nl $\maxs(T)\;\leftarrow \; \hat{x}_{P}+ \hat{y}_{P}+\tot^{+}(T)$ ; \label{line:compute_maxs1}
\nl $\sets(T) \;\leftarrow \setl(A_{L})\cup \setr(A_{R})\cup \{A_{j}\;:\;j\neq \;L, R, \tot(A_{j})> 0\}$;
}{
\nl $\maxs(T)\;\leftarrow \;z_{P}$; $\sets(T)\;\leftarrow \; \sets(A_{s})$ \label{line:compute_maxs2}
}
\caption{Computation of the variables when the root is a $P$-node} \label{P_interval}
\end{algorithm}

Second, suppose that the elements of $\sets(T)$ come from at least two sets. Using similar arguments as $\maxl(T)$, we can show that the PQ-convex is the union of: a set $\setl(A_{\ell})$ which is on the right of the PQ-convex set (in other words, a suffix), a set $\setr(A_{r})$, $r \neq \ell$ which is on the left of the PQ-convex set (a prefix), and all other sets $A_{j}$, $j \neq \ell, r$ which satisfies $\tot( A_{j})> 0$. This scenario is depicted in Figure~\ref{fig:Pnode}.
\begin{equation}\label{eq4}
   \textstyle \maxs(T)=\maxl(A_{\ell})+\maxr(A_{r})+\sum\limits_{j\neq \ell,r ~;~ {\scriptsize \tot( A_{j})}> 0} \tot(A_{j})
\end{equation}
\indent In this case the idea is to find the pair ($A_{\ell},A_{r})$,$\ell\neq r$ which maximizes the weight of the PQ-convex. We can achieve this operation in time $O(k^2)$:
\begin{itemize}
    \item[-] compute the collection of sets $A_i$ which verify $\tot(A_i) \ge 0$ and determine the sum $\tau$ of their total weights (linear time) 
    \item[-] enumerate all possible pairs ($A_{\ell},A_{r})$,$1\le \ell\neq r \le k$ (quadratic time) and compute the weight of the convex set obtained for this pair suffix-prefix. Deduce the value of the middle sets by reducing from $\tau$, if necessary, the weights of $A_{\ell}$ and $A_r$ (if their total is positive).
\end{itemize}

Eventually we compare the maximum values obtained in both cases (Equations~\eqref{eq3} and~\eqref{eq4}) and take the overall maximum. This is described in lines \ref{line:enum_pair}-\ref{line:compute_maxs2} of Algorithm \ref{P_interval}. This process can be achieved in quadratic time.\\

$\bullet$ \textit{$Q$-node.} We suppose now that the root of $T$ is a $Q$-node and the variables of the subtrees below this root have been computed.

For the computation of $\maxl(T)$, we distinguish two cases. Indeed set $\setl(T)$ appears left when the ordering is either $A_{1},...,A_{k}$ or $A_{k},...,A_{1}$. First, we begin with the increasing order  $A_{1},...,A_{k}$. Let $A_{i}$ the rightmost set that contains elements of $\setl(T)$. Then we have,  
\begin{equation}\label{eq5}
  \textstyle \maxl(T)= \sum\limits_{j=1}^{i-1} \tot(A_{j}) + \maxl(A_{i})
\end{equation}

With this ordering, value $\maxl(T)$ is thus obtained by listing all $A_{i}$  and determining the one  which maximizes the right-hand side of Equation~\eqref{eq5}. Second with the decreasing order $A_{k},...,A_{1}$, set $\setl(T)$ is the union of set $\setl(A_{i})$ with all sets $A_{j}$, where $j>i$.

\begin{equation}\label{eq6}
  \textstyle \maxl(T)= \maxl(A_{i})+\sum\limits_{j=i+1}^{k} \tot(A_{j})
\end{equation}
\indent Eventually, we compute the maximum of the two values obtained for each case and obtain $\maxl(T)$. This operation takes $O(k)$ since we need to explore all sets $A_i$ and update the total weight of all other members of the convex set increasingly. Computing $\maxr(T)$ is symmetric.

For the computation of  $\maxs(T)$, we distinguish three cases. First, if the elements of  $\sets(T)$ are all contained in some set $A_{i}$, then  $\maxs(T)$= $\maxs(A_{i})$. Second, suppose that $\sets(T)$ contains elements of at least two sets and that this  PQ-convex appears when the subtrees are in ascending order $A_{1},...,A_{k}$. Let $A_{\ell}$ the leftmost subtree having a nonempty intersection with  $\sets(T)$, $A_{r}$ the rightmost. Value $\maxs(T)$ can be thus written : 
 \begin{equation}\label{eq7}
  \textstyle \maxs(T)= \maxr(A_{\ell})+ \maxl(A_{r})+\sum\limits_{j=\ell+1}^{r-1} \tot(A_{j})
\end{equation}
Observe that this operation can be naturally achieved in time $O(k^2)$, by enumerating all possible prefix $A_r$, then by increasing parameter $j$ one by one, and updating the total of middle sets while adding to it value $\maxl(A_j)$. At the end, we return the maximum value obtained for the weight of a convex set. Note that it could certainly be improved to a linear running time by using dynamic programming, but since quadratic time was already needed for a previous case, we keep this algorithmic version.

Finally, the third case is when $\sets(T)$ contains elements of several sets and appears with the subtrees in the decreasing order $A_{k},...,A_{1}$. Let $A_{r}$ the leftmost subtree sharing elements with $\sets(T)$, $A_{\ell}$ the rightmost.

 \begin{equation}\label{eq8}
  \textstyle \maxs(T)= \maxl(A_{\ell})+ \maxr(A_{r})+\sum\limits_{j=\ell+1}^{r-1} \tot(A_{j})
\end{equation}

The processing time to achieve this computation is identical to the one for the ascending order. In summary, to obtain  $\maxs(T)$, we determine the maximum of the three cases. The overall running time is thus $O(k^2)$.

In summary, the time needed to determine the four variables for each subtree $T$ rooted in $v$ is $O(k^2)$, where $k$ is the number of children of $v$. In summary, the time needed for the whole recursion process is thus $O(n^2)$. Since vertices with one child do not provide any particular information, we can restrict ourselves to irreducible trees where the number of leaves is linear in $n$.
\end{proof}

There is no doubt, in our opinion, that such a dynamic programming algorithm can be adapted to other types of nodes, as C and R. However, since the case of PQ-trees is the most crucial for applications, we do not focus on such technical details for orher kinds of typed-laminar trees.

%
%
%
%
%

\section{Interval graphs}
 A graph $G=(V,E)$ is an \emph{interval graph} iff it models a set of intervals on the real line such that two vertices are adjacent iff the corresponding intervals have a non empty intersection. Hajos~\cite{Hajos57} was the first one to mention these graphs in the literature. Interval graphs are related to problems in biology, psychology and traffic light sequencing.
For interval graphs, the natural graph-convexity space to consider is the monophonic convexity since they are chordal.
If an interval representation of the graph is given, a simple scan with a sweep line allows to compute the maximum-weight convex set which is an interval of this representation. In fact the problem is reduced to the \textsc{maximum sum subsequence} problem already discussed. Otherwise, we need to deal with PQ-trees a structure that allows to  generate all interval representations.
It is well-known that for every interval graph $G$, one can compute a PQ-tree $T_G$ whose leaves are in bijection with the maximal cliques of $G$ with only two types of internal nodes $P, Q$. Hence, we are going to transform our graph problem into a problem on PQ-trees.

Using results in \cite{Jesse}, we will prove that a given interval convexity is a convex geometry.

\begin{definition}[Interval convexity space~\cite{Jesse}] Given an interval space $(V, I)$, the interval convexity induced by $I$
on $V$ is defined as follows: a subset $ C \subseteq V$ is interval convex if $I[a, b] \subseteq C$ for all
$a, b \in C$. An interval space $(V,I)$ is called geometric interval if the operator $I$ is geometric.
\end{definition}

Chv\'{a}tal proposed an interesting property on interval spaces $(V, I)$.

\begin{property}[Strong Chv\'{a}tal Property]
For all $a, b, c \in V$ and $y \in I[b, c]$ and $z \in I[a, y]$, it
holds that $z \in I[a, b]$ or $z \in I[a, c]$.
\end{property}

The Chv\'{a}tal Property is very useful in the study of interval convexities.

\begin{definition}[Interception convexity]
Let $G = (V,E)$ be a graph and $z, a, b \in V$. We say that $z$ is in the
\emph{interception interval} of $a$ and $b$, denoted as $z \in I_{int}[a, b]$ iff there exists vertices
$u, v\in N(z)$ such that $u$ and $v$ are in different connected components of $G-(N[z]\backslash \{u, v\})$,
where $a$ (resp. $b$) is in the same component as $u$ (resp. $v$). We call $u$ the \emph{witness} of $a$ and $v$ the witness of $b$ with regard to $z \in  I_{int} [a,  b]$. The interval convexity
$(V, C_{int})$ induced by this operator is called the \emph{interception convexity}.
\end{definition}

\begin{lemma}[\cite{Jesse}] Let $G = (V,E)$ be a graph and let $(V, C_{int})$ be its interception convexity.
If $z \in I_{int}[a, b]$, then $z$ intercepts the interior of any path between $a$ and $b$.
\end{lemma}

The interception convexity implies a linear vertex ordering of a graph.

\begin{definition}[Interception convexity order] Let $G = (V,E)$  be a connected graph and let $\sigma = (v_{1}, . . . , v_{n})$ be a
linear order of its vertices. We say that $\sigma$ is an \emph{interception convexity order} if $(v_{1}, . . . , v_{i})$ is an interception convex set for every $i \in \{1, . . . ,n\}$.
\end{definition}

\begin{theorem}[\cite{Jesse}]
\label{convex geometry} For any graph $G = (V,E)$ the following properties are equivalent:
\begin{enumerate}
\item The interception interval operator $I_{int}$ of $G$ fulfils the Strong Chv\'{a}tal Property;
\item The interception convexity of $G$ is a convex geometry;
\item $G$ possesses a interception-convexity order.
\end{enumerate}
\end{theorem}

This leads to the following proposition.

\begin{proposition}\label{Interval-interval}
Every convex set of an interval graph $G$ can be associated to an interval in some interval representation of $G$.
\end{proposition}
\begin{proof} Let $G=(V,E)$ be an interval graph, $A$ a subset of $V$ and $A'=\{z : \exists a,t \in A ~\mbox{with}~ z \in I_{int}[a,t]\}$. As $\conv(A)$ is a superset of $A$, it is a superset of $A \cup A'$: let us prove that $A \cup A'$ is a convex set. Consider an arbitrary $z \in V$ such that  $z \in I_{int}[a,t]$ for some $a,t \in A \cup A'$, we are going to prove that $z \in A \cup A'$.

\textbf{Case 1: } $a, t \in A $, in this case $z \in A' $, by definition of $A'$.

\textbf{Case 2: } $a \in A, t \in A'$. By definition of $A'$ there are $b,c \in A$ with $t \in I_{int}[b,c]$, if we refer to Theorem~\ref{convex geometry}, by Chv\'atal property with $t=y$, we have $z \in I_{int}[a,b]$ or $z \in I_{int}[a,c]$ then $z \in A'$.

\textbf{Case 3: } $a, t \in A'$. There are $b,c \in A$, $z \in I_{int}[a,b]$ or $z \in I_{int} [a,c]$ then we are back in case 2.
\end{proof}

In other words the above Proposition \ref{Interval-interval} yields that the m-convexity has Carathéodory number 2 for interval graphs.

\begin{lemma} \label{modules}
Let $M$ be a  module of $G$  and $S$ be a geodesic convex set (resp. m-convex, m3-convex) set. If for two non adjacent vertices $a,b \notin M$, there exists $x \in M$ that belongs to $I[a,b]$, then every vertex of $M$ belongs to $I[a,b]$.
\end{lemma}
\begin{proof}
Let us consider the m-convexity case and suppose there exists a chordless path in $G$:
$\mu=[a, \dots u,x,v\dots b]$ with $u,v \notin M$ and $x \in M$.
First we notice that $\mu \cap M =\{x\}$, else $\mu$ would not be chordless.
Using the module definition for every $z \in M$ there exists a chordless path
in $G$:
$\mu_z=[a, \dots u,z,v\dots b]$ and therefore $z \in I[a,b]$. The proof is similar for the two other convexities.
\end{proof}

\begin{corollary}\label{reductionlemma}
For a module $M$ in $G$, in order to compute the intervals $I[a, b]$ with $a,b \notin M$, in geodesic convexity space (resp. m-convexity, m$_3$-convexity), it suffices to compute them in a graph $G'$ obtained by  contracting $M$ into a single vertex $\cal M$ with weight equal to the sum of the weights of the vertices of $M$.
\end{corollary}

Any legitimate plane drawing of the PQ-tree $T$ of an interval graph $G$ yields an ordering of the maximal cliques in which the cliques which contain a given vertex are consecutive.
Legitimate here means that for a $Q$-node only two dual orders of its children are allowed, and for a $P$-node every order of its children are possible.
Such a total ordering of the maximal cliques
$C_1, \dots, C_k$ yields an interval representation as follows: to each vertex of $G$ we associate the unique interval of the cliques that contain it.
So $T$ is a short description  of all interval representations of $G$.
 
\begin{theorem}\label{properinterval}
The maximum convex set problem can be solved in linear time for proper interval graphs.
\end{theorem}

\begin{proof}
For a prime proper interval graph, there is a unique interval representation.
It is well-known that for a proper interval graph, the only possible  modules are sets of true twins. Using Lemma\ref{modules} and its corollary 
we can represent all the interval representation using a typed-laminar tree of height $2$. The root is a Q-node and if there exist some sets of true twins we add a level of P-nodes.

Whence the maximum convex set can be computed in linear time in such a laminar tree.
\end{proof}

\begin{theorem}\label{intervalgr}
The maximum convex set problem can be solved in $O(n^2)$ for interval graphs.
\end{theorem}
\begin{proof}
Let us consider an interval graph $G$. First  compute the PQ-tree of $G$ which leaves are in bijection with the maximal cliques of $G$, which can be done in linear time~\cite{BoothL76}.
As explained in \cite{Crespelle07}  we can derive from this tree in linear time its modular decomposition tree. With these two trees we can build a typed-laminar tree on $V(G)$.
 
Q-nodes (resp. P-nodes)  of the PQ-tree correspond to prime nodes (resp. to series or parallel nodes) in the modular decomposition tree. As interval graphs Prime nodes have only one interval representation and its dual. Therefore for each prime node we keep the ordering of its children via the Q-node.

This gives a typed-laminar tree with two types of Q (or prime) nodes and P nodes and the leaves are in bijection with $V(G)$. Furthermore this tree encodes all interval representations of $G$.
To finish the proof we use  Proposition \ref{Interval-interval} and Theorem \ref{PQ-tree}  to obtain an $O(n^2)$ algorithm.
\end{proof}

\section{Split and threshold graphs}

A graph $G=(K \cup I,E)$ is a split graph if the vertex set can be partitioned into two parts $K, I$ such that $G[K]$ is a complete graph and $G[I]$ is an independent set.
Between $K$ and $I$ we can have any bipartite graph.
A convex set in $G$ is therefore a pair $C=(K', I')$ with $K' \subseteq K$ and $I' \subseteq I$ and closed under the monophonic convexity.

Let us mention some simplifications rules on our problem. Every vertex in $I$ is a simplicial vertex of $G$ so it is an extreme vertex for the monophonic convexity, therefore:

\textbf{Fact 1}: For the computation of the maximum convex set in $G$, we can only keep the vertices in $I$ with positive weight.

Similarly we can reduce the independent twins, {\em i.e.} that have exactly the same neighbourhood in $K$.

\textbf{Fact 2}: let $i,j$ be two positive twins in $I$, for the computation of the maximum convex set in $G$, we can replace these 2 vertices by a unique one denoted by $ij$ with $\omega(ij) =\omega(i)+\omega(j)$.

Furthermore, since any subset of $K$ is a convex set,the subset of all positive vertices of $K$ yields a bound to the maximum weighted convex set in $G$.

\begin{lemma}
Let $i,j \in I$, then $N(i) \cup N(j) \subseteq \conv(\{i,j\})$.
\end{lemma}


We refer to the \emph{implicit representation} of a split graph as the bipartite graph between $K$ and $I$. From now on, let $m$ denote the number of edges of the implicit representation of $G$.

A particular case of split graphs are the threshold graphs for which the bipartite between $K$ and $I$ does not contain any pair of independent edges ($K_2 + K_2$).
In other words the neighbourhoods of the independent vertices are totally ordered.
We will now show that a linear time algorithm for this particular class of split graphs can be obtained.

We order $I$ as $i_1, \dots i_p$ such that
$N(i_1) \subseteq N(i_2) \dots \subseteq N(i_p)$.

\begin{algorithm}[h!]

\KwIn{A threshold graph $G=(K, I)$ given by its implicit representation and with its independent vertices totally ordered by inclusion of their neighbourhoods.}
\KwOut{The value of a maximum convex set in G.}
\BlankLine
\Begin{
     
     $K^+ \leftarrow \sum\limits_{x \in K, \omega(x)> 0}\omega(x)$\;
$\maxs \leftarrow K^+$\;
$W \leftarrow 0$\;
     \For {$j=1$ to $p$}
     {$s^+(j)=\leftarrow \sum\limits_{x \in N(j), \omega(x)> 0}\omega(x)$\;
$s^-(j)=\leftarrow \sum\limits_{x \in N(j), \omega(x)< 0}\omega(x)$\;
$s(j) \leftarrow s^+(j)+s^{-}(j)$\;

$W \leftarrow W + \omega(j)$\;
$\maxs \leftarrow \max \{\maxs, \omega(j) + s^+(j),  W + K^+ + s(j) \}$ \;
$K^+ \leftarrow K^+ - s^+(j)$ \;

}
}
\caption{Maximum weight convex set for threshold graphs } \label{threshold}
\end{algorithm}

\begin{theorem}\label{Threshold} Algorithm \ref{threshold}
computes the maximum weight of a convex set in $G$ in $O(n)$.
\end{theorem}
\begin{proof}
The easy part is the complexity evaluation since every independent vertex is considered only once so as its neighbourhood, so in the whole $O(|I|+m)$.
Furthermore the initialization lines before the For loop are in $O(|K|)$.
Let us now consider the proof of correctness.
First we notice that $W$ is the sum of the weight of the independent up to $j$. Similarly $K^+$ is the sum of the positive weights in $K \setminus N(j)$.
Every value considered in the above algorithm and compared to $\maxs$ corresponds to the weight of some convex set in $G$. The only thing to prove is that we do not miss any convex set during the algorithm.

Let $S$ be a convex set in $G$. If $S$ does not contain any independent vertex, then $S$ has been considered during the initialization instructions.

Else let $i_k$ be the maximal index with respect to the ordering of the independent vertices present in $S$.
If $i_k$ is the only one in $S$ then the weight of $S$ has been considered in the second term of the max instruction  with $j=i_k$.
Else necessarily $N(i_k)\subseteq S$ and $S$ has been considered in the third term of the max instruction with $j=i_k$.
\end{proof}

Of course the above algorithm only computes the value of a maximum weight of a convex set,
but an easy modification can keep track of the convex set when the value $\maxs$ is modified. This can be done within the same complexity. 
 
\section{Perspectives}

Our immediate perspective  is to improve the algorithms for the known  polynomial cases, as for example  for splits graphs, chordal graphs and Ptolemaic graphs.
PCR trees were used for planarity testing and circular arc graphs recognition, could we  generalize our work for \textsc{max weight convex set} on these classes of graphs?
Next is the question of NP-completeness border: Is \textsc{max weight convex set} polynomial for cocomparability graphs and AT-free graphs on their respective convexities?
It should be the case for cocomparability graphs which are very closed to interval graphs, as explained in \cite{DusartHC24}, but not so clear for AT-free graphs \cite{corneil1997asteroidal}.
Similarly, it could be very interesting to characterize the interval graph-convexities  that behave nicely with modules as in Lemma \ref{modules}.

A more fundamental question arises from Carathéodory numbers. One may conjecture that for a geometric convexity with bounded Carathéodory number, \textsc{max weight convex set} is polynomial. Unfortunately it is false using Carathéodory theorem (1907) bounding this number for geometric convexity in $\mathbb{R}^p$ by $p+1$. As we already mentioned it, \textsc{max weight convex set} is NP-hard for $\mathbb{R}^3$. Hence, we wonder what further properties on geometric convexities could ensure polynomiality.

\subsection*{Acknowledgments}

This work was supported by the ANR (Agence Nationale de la Recherche) research project SPAWN (ANR-25-CE48-1750).

 \bibliographystyle{plain}
\bibliography{Biblio}

@string{order = {Order}}

@string{siam-dm = {SIAM J. D.M.}}

@article{IBARRA2009,
author={Louis Ibarra},
title = {The clique-separator graph for chordal graphs},
journal = {Discrete Applied Mathematics},
volume = {157},
number = {8},
pages = {1737-1749},
year = {2009},
issn = {0166-218X},
doi = {https://doi.org/10.1016/j.dam.2009.02.006},
}

@article{El-KebirK14,
  author       = {Mohammed El{-}Kebir and
                  Gunnar W. Klau},
  title        = {Solving the Maximum-Weight Connected Subgraph Problem to Optimality},
  journal      = {CoRR},
  volume       = {abs/1409.5308},
  year         = {2014},
  url          = {http://arxiv.org/abs/1409.5308},
  eprinttype    = {arXiv},
  eprint       = {1409.5308},
  }

@phdthesis{Jesse,
    author = {Jesse Beisegel},
    title = {Convexity in Graphs: vertex order characterization and graph searching},
    school = {Technischen Universit¨at Cottbus},
    year = {2019}
}

@article{BoothL76,
  author       = {Kellogg S. Booth and
                  George S. Lueker},
  title        = {Testing for the Consecutive Ones Property, Interval Graphs, and Graph
                  Planarity Using {PQ}-Tree Algorithms},
  journal      = {J. Comput. Syst. Sci.},
  volume       = {13},
  number       = {3},
  pages        = {335--379},
  year         = {1976},
  }

@book{Schrijver,
author={Alexander Schrijver},
title={Combinatorial Optimization},
year={2002},
publisher={Springer-verlag},
series={Algorithms and Combinatorics vol. 24},
}

@article{DusartHC24,
  author       = {J{\'{e}}r{\'{e}}mie Dusart and
                  Michel Habib and
                  Derek G. Corneil},
  title        = {Maximal Cliques Lattices Structures for Cocomparability Graphs with
                  Algorithmic Applications},
  journal      = {Order},
  volume       = {41},
  number       = {1},
  pages        = {99--133},
  year         = {2024},
  url          = {https://doi.org/10.1007/s11083-023-09641-x},
  doi          = {10.1007/S11083-023-09641-X},
 
}

@inproceedings{CarlsonE06,
  author       = {Josiah Carlson and
                  David Eppstein},
  title        = {Procs. of {SWAT}},
  series       = {Lecture Notes in Computer Science},
  volume       = {4059},
  pages        = {400--410},
  publisher    = {Springer},
  year         = {2006}
}

@book{Bentley84,
author={Jon Bentley},
title={Programming Pearls, Second Edition},
publisher={Addison-Welsey},

year={2000},
}

@article{Picard75,
author={J.-C. Picard},
title={Maximal closure of a graph and applications to combinatorial problems},
Journal={Management Science},
volume={22(11)},
pages={1268–1272},
year={1975},
}

@phdthesis{Merckx19,
    author ={Keno Merckx},
    title = {Optimization and Realizability Problems for Convex Geometries},
    school ={University of Louvain ULB} ,
    year = {2019}
}

@article{MerckxCD16,
  author       = {Keno Merckx and
                  Jean Cardinal and
                  Jean{-}Paul Doignon},
  title        = {On the shelling antimatroids of split graphs},
  journal      = {Electron. Notes Discret. Math.},
  volume       = {55},
  pages        = {199--202},
  year         = {2016}
}

@article{CardinalDM17,
  author       = {Jean Cardinal and
                  Jean{-}Paul Doignon and
                  Keno Merckx},
  title        = {On the shelling antimatroids of split graphs},
  journal      = {Discret. Math. Theor. Comput. Sci.},
  volume       = {19},
  number       = {1},
  year         = {2017}
}

@article{CardinalDM19,
  author       = {Jean Cardinal and
                  Jean{-}Paul Doignon and
                  Keno Merckx},
  title        = {Finding a Maximum-Weight Convex Set in a Chordal Graph},
  journal      = {J. Graph Algorithms Appl.},
  volume       = {23},
  number       = {2},
  pages        = {167--190},
  year         = {2019}
  
}

@article{Howorka81,
  author       = {Edward Howorka},
  title        = {A characterization of ptolemaic graphs},
  journal      = {J. Graph Theory},
  volume       = {5},
  number       = {3},
  pages        = {323--331},
  year         = {1981}
  
}

@article{Hajos57, 
author={G. Hajos},
title={{\"Uber eine Art von Graphen}},
journal={Internationale Math. Nachrichten},
volume={11},
year={1957},

}

@article{edelman1985theory,
  title={The theory of convex geometries},
  author={Edelman, P. and Jamison, R.},
  journal={Geometriae dedicata},
  volume={19},
  number={3},
  pages={247--270},
  year={1985},
  publisher={Springer}
}

@article{shih1999new,
  title={A new planarity test},
  author={Shih, Wei-Kuan and Hsu, Wen-Lian},
  journal={Theoretical Computer Science},
  volume={223},
  number={1},
  pages={179--192},
  year={1999},
  publisher={Amsterdam: North-Holland Pub. Co., 1975-}
}

@article{corneil1997asteroidal,
  title={Asteroidal triple-free graphs},
  author={D. Corneil and S. Olariu and L. Stewart},
Journal = siam-dm,
  volume={10},
  number={3},
  pages={399--430},
  year={1997},
  publisher={SIAM}
}

@inproceedings{Chvatal08,
  author    = {V. Chv{\'{a}}tal},
  title     = {Antimatroids, Betweenness, Convexity},
  booktitle = {Research Trends in Combinatorial Optimization},
  pages     = {57--64},
  year      = {2008},
  }

@article{Duchet88,
  author    = {P. Duchet},
  title     = {Convex sets in graphs, {II.} Minimal path convexity},
  journal   = {J. Combinatorial Theory B},
  volume    = {44},
  number    = {3},
  pages     = {307--316},
  year      = {1988},
  }

@article{FarberJ87,
  author    = {M. Farber and
               R. E. Jamison},
  title     = {On local convexity in graphs},
  journal   = {Disc. Math.},
  volume    = {66},
  number    = {3},
  pages     = {231--247},
  year      = {1987},
  
}

@book{korte2012greedoids,
  title={Greedoids},
  author={B. Korte and L. Lov{\'a}sz and R. Schrader},
  year={2012},
  volume={4},
  publisher={Springer}
}

@article{menger1928untersuchungen,
  title={Untersuchungen {\"u}ber allgemeine Metrik},
  author={K. Menger},
  journal={Mathematische Annalen},
  volume={100},
  number={1},
  pages={75--163},
  year={1928},
  publisher={Springer}
}

@article{abam2022maximum,
  title={Maximum Weight Convex Polytope},
  author={Abam, Mohammad Ali and Lavasani, Ali Mohammad and Pankratov, Denis},
  journal={arXiv preprint arXiv:2207.12915},
  year={2022}
}

@article{eppstein1992finding,
  title={Finding minimum area k-gons},
  author={Eppstein, David and Overmars, Mark and Rote, G{\"u}nter and Woeginger, Gerhard},
  journal={Discrete \& Computational Geometry},
  volume={7},
  pages={45--58},
  year={1992},
  publisher={Springer}
}

@article{DraganNB99,
  author       = {Feodor F. Dragan and
                  Falk Nicolai and
                  Andreas Brandst{\"{a}}dt},
  title        = {Convexity and HHD-Free Graphs},
  journal      = {{SIAM} J. Discret. Math.},
  volume       = {12},
  number       = {1},
  pages        = {119--135},
  year         = {1999},
  url          = {https://doi.org/10.1137/S0895480195321718},
  doi          = {10.1137/S0895480195321718},
  
}

@misc{chen2022maximumflowminimumcostflow,
      title={Maximum Flow and Minimum-Cost Flow in Almost-Linear Time}, 
      author={Li Chen and Rasmus Kyng and Yang P. Liu and Richard Peng and Maximilian Probst Gutenberg and Sushant Sachdeva},
      year={2022},
      eprint={2203.00671},
      archivePrefix={arXiv}
}

@article{KW71,
    author = {D.C. Kay and E.W. Womble},

    title = {Axiomatic convexity theory and relationships
between the {Carath\'eodory, Helly, and Radon} numbers},
    journal = {Pacific
Journal of Mathematic},
    year = {1971},
volume={38},
pages={471-485},
}

@article{Doignon16,
  author       = {Jean{-}Paul Doignon},
  title        = {A Convex Polytope and an Antimatroid for any Given, Finite Group},
  journal      = {Electron. Notes Discret. Math.},
  volume       = {54},
  pages        = {21--25},
  year         = {2016}
}

@phdthesis{Crespelle07,
  author       = {Christophe Crespelle},
  title        = {Repr{\'{e}}sentations dynamiques de graphes. (Dynamic Graph Representations)},
  school       = {Montpellier 2 University, France},
  year         = {2007}
}

@book{cosi25,
  title={{Optimization and Information Systems}},
  editor = {{Engelbert Mephu Nguifo}, {Abderrazak Sebaa}},
  author={{17th Conference, COSI 2025, Bejaia, Algeria, June 1–3, 2025, Proceedings}},
  year={2026},
  publisher={Springer Cham}
}

 \end{document}